\documentclass[preprint]{elsarticle}
\usepackage[utf8x]{inputenc}
\usepackage{amsmath, amssymb}
\usepackage{mathenv,amsthm}
\usepackage[boxed,commentsnumbered]{algorithm2e}
\usepackage{graphicx}
\usepackage{multirow}
\usepackage[toc,page]{appendix} 
\usepackage{fmtcount}
\usepackage{array}

\usepackage[margin=3.5cm]{geometry}

\newtheorem{theorem}{Theorem}[section]
\newtheorem{definition}{Definition}[section]
\newtheorem{lemma}[theorem]{Lemma}

\newenvironment{remark}[1][Remark :]{\begin{trivlist}
\item[\hskip \labelsep {\bfseries #1}]}{\end{trivlist}}

\newcommand{\F}{\mathbb{F}}

\renewcommand{\o}{\mathcal{O}}

\journal{Theoretical Computer Science}
\journal{}

\begin{document}

\begin{frontmatter}

\title{Further Refinements of Miller Algorithm \\on Edwards curves}

\author{Duc-Phong Le\corref{cor1}}
\ead{tslld@nus.edu.sg}

\author{Chik How Tan}
\ead{tsltch@nus.edu.sg}

\cortext[cor1]{Corresponding author}

\address{Temasek Laboratories, National University of Singapore\\5A Engineering Drive 1, \#09-02, Singapore 117411.}

\begin{abstract}

Recently, Edwards curves have received a lot of attention in the cryptographic community due to their fast scalar multiplication algorithms. Then, many works on the application of these curves to pairing-based cryptography have been introduced. Xu and Lin (CT-RSA, 2010) presented refinements to improve the Miller algorithm that is central role compute pairings on Edwards curves.
In this paper, we study further refinements to Miller algorithm. Our approach is {\em generic}, hence it allow to compute {\em both} Weil and Tate pairings on pairing-friendly Edwards curves of {\em any} embedding degree. We analyze and show that our algorithm is faster than the original Miller algorithm and the Xu-Lin's refinements. 

\end{abstract}


\end{frontmatter}

\section{Introduction}
 
In 2007, Bernstein and Lange~\cite{BL07} introduced Edwards curves to cryptography. Their study and subsequent works in~\cite{BL07b,BBJ+08,HWC+08,BJ11} showed that the addition law on such curves is more efficient than all previously known formulas. Edwards curves have thus attracted great interest in applications that require elliptic curve operations to achieve faster arithmetic. 
Then, the application of Edwards curves to pairing-based cryptography has been studied in several research  papers~\cite{DS08,IJ08,ALNR10,XL10,LT13}. Although, pairing computation on Edwards curves is slightly slower than on Weierstrass curves so far. However, in many pairing-based cryptosystems, the most time-consuming operation is still to compute scalar multiples $aP$ of a point $P$. 

Pairing (or {\it bilinear map}) is probably the most useful cryptographic tool in the 2000s. It was first introduced to cryptography in Joux's seminal paper~\cite{Jou00} in 2010 that describes a tripartite (bilinear) Diffie-Hellman key exchange. Then, the use of cryptosystems based on pairings has had a huge success with some notable breakthroughs such as the first identity-based encryption scheme~\cite{BF01}, the short signature scheme~\cite{BLS01}.

Ever since it was first described, Miller's algorithm~\cite{Mil86,Mil04} has been the heart of the computation of pairings on elliptic curves. 
Many papers are devoted to improvements in its efficiency. For example, it can run faster on pairing-friendly elliptic curves that belong to specific families~\cite{BKLS02, BLS03,BN05,CLN10}. Another approach of improving the Miller's algorithm is to reduce the Miller-loop length by introducing variants of Tate pairings, for example Eta pairing~\cite{BGHS07}, Ate pairing~\cite{HSV06}, and particular {\it optimal pairings}~\cite{Ver10,Hes08}. For a more generic approach, studies in~\cite{BMX06,BELL10,LL11} improved the performance for computing pairings of any type ({\it i.e.}, Weil, Tate, optimal pairings), and on {\it generic} pairing-friendly elliptic curves. 

Basically, Miller's algorithm is based on a rational function $g$ of three points $P_1, P_2, P_3$. This function is called Miller function and has its divisor $div(g) = (P_1) + (P_2) - (P_3) - (\o)$, where $\o$ is a distinguished rational point. For curves of Weierstrass form, this function is defined to be the line passing through points $P_1$ and $P_2$ divided by the vertical line passing through the point $P_3$, where $P_3= (P_1 + P2)$. 
On Edwards curves, finding such a point $P_3$ is not straightforward  as in Weierstrass curves because Edwards equation has degree $4$, {\em i.e.} any line has $4$ intersections with the curves instead of $3$ on Weierstrass curves.


In~\cite{ALNR10}, Arene {\em et al.} presented the first geometric interpretation of the group law on Edwards curves and showed how to compute Tate pairing on twisted Edwards curves by using a conic $\cal C$ of degree $2$. They also introduced explicit formulas with a focus on curves having an even {\em embedding degree}\footnote{Let $E$ be an elliptic curve defined over a prime finite field $\F_p$, and $r$ be a prime dividing $\#E(\F_p$). The embedding degree of $E$ with respect to $r$ is the smallest positive integer $k$ such that $r | p^k - 1$. In other words, $k$ is the smallest integer such that $\F^*_{p^k}$ contains $r$-roots of unity.}. 
In order to speed up the pairing computation on {\em generic} Edwards curves, Xu and Lin~\cite{XL10}  proposed refinements to Miller's algorithm. Their refinements are inspired from Blake-Murty-Xu's refinements on Weierstrass curves~\cite{BMX06}. Their refinements are {\it generally} faster than the original Miller's algorithm on Edwards curves described in~\cite{ALNR10}.

In this paper, we further extend the Blake-Murty-Xu's method on Edwards curves and propose new refinements of Miller's algorithm. Similarly Xu-Lin's refinements, our approach is {\it generic}. Although it did not bring a dramatic efficiency as that of Arene {\em et al.} for computing Tate pairing over Edwards curves with {\it even} embedded degree, but the proposed refinements can be used to compute pairing of {\it any} type over pairing-friendly Edwards curves with {\it any} embedding degree. 
This approach is of particular interest to compute optimal pairings~\cite{Ver10, Hes08}, and in situations where the denominator elimination technique using a twist is not possible (e.g., Edwards curves with {\it odd} embedding degrees).\footnote{Note that by definition optimal pairings only require about $\log_2(r)/\varphi(k)$ iterations of the basic loop, where $r$ is the group order, $\varphi$ is Euler's totient function, and $k$ is the embedding degree. For example, when $k$ is prime, then $\varphi(k) = k - 1$. If we choose a curve having embedding degree $k \pm 1$, then $\varphi(k\pm 1)\leq \frac{k+1}{2}$ which is roughly $\frac{\varphi(k)}{2}=\frac{k-1}{2}$, so that at least twice as many iterations are necessary if curves with embedding degrees $k \pm 1$ are used instead of curves of embedding degree $k$.}
We also analyze and show that our new algorithm is faster than the original Miller's algorithm and its refinements presented in~\cite{XL10}.

The paper is organized as follows. In Section 2 we briefly recall some background on pairing computation on Edwards curves, and then the Xu-Lin's refinements. In Section 3 we present our refinements of Miller's algorithm. Section 4 gives some discussion on performance of the proposed algorithms. Section 5 is our conclusion.

\section{Preliminaries}

\subsection{Pairings on Edwards Curves}
\label{sec:pair_comp}

Let $\F_p$ be a prime finite field, 
where $p$ is a prime different from $2$. A {\em twisted Edwards curve} $E_{a, d}$ defined over $\F_p$ is the set of solutions $(x, y)$ of the following {\em affine equation}:

\begin{equation}
\label{eq:ell}
 E_{a, d} : ax^2 + y^2 =  1 + dx^2y^2,
\end{equation}

\noindent where $a, d \in \F_p^*$, and $a \ne d$. Edwards curves are special case of twisted Edwards curves where $a$ can be rescaled to $1$. Twisted Edwards curves have the fastest doubling and addition operations in elliptic curve cryptography (see~\cite{BBJ+08}).

Cryptographic pairing is a bilinear map that takes as input two points on elliptic curves defined over finite fields and returns a value in the extension finite field. The key to the definition of pairings is the evaluation of rational functions in divisors. The pairings over (hyper-)elliptic curves are computed using the algorithm proposed by Miller \cite{Mil04}. The main part of Miller's algorithm is to construct the rational function $f_{r,P}$ and evaluating $f_{r,P}(Q)$ with $div(f_{r,P}) = r(P) - (rP) - [r - 1](\cal{O})$ for divisors $P$ and $Q$. In this section, we just recall the Miller algorithm that is so far the best known method to compute pairings. Readers who want to study more about pairings can take a look at papers~\cite{Mil04, Ver10}.

Let $m$ and $n$ be two integers, and $g_{mP,nP}$ be a rational function whose divisor $div(g_{mP,nP}) = (mP) + (nP) - ([m + n]P) - (\mathcal{O})$. We call the function $g_{mP, nP}$ {\em Miller function}. Miller's algorithm is based on the following lemma.

\begin{lemma}[Lemma 2, \cite{Mil04}]
\label{Millerlemma}
For $n$ and $m$ two integers, up to a multiplicative constant, we have
\begin{equation}
\label{eq:Miller}
 f_{m + n, P}=f_{m, P} f_{n, P} g_{mP, nP}.
\end{equation}

\end{lemma}

Equation~(\ref{eq:Miller}) is called \emph{Miller relation}, which is proved by considering divisors. The Miller algorithm makes use of Lemma~\ref{Millerlemma} with $m = n$ in a doubling step and $n = 1$ in an addition step. For Edwards curves, Arene {\em et al.}~\cite{ALNR10} defined Miller's function in the following theorem.



\begin{theorem}[Theorem 2, \cite{ALNR10}]\label{theo1}
Let $a, d \in \F_p^*, a \ne d$ and $E_{a, d}$ be a twisted Edwards curve over $\F_p$. Let $P_1, P_2 \in E_{a, d}(\F_p)$. Define $P_3 = P_1 + P_2$. Let  $\phi$ is the equation of the conic $\cal C$ passing through $P_1, P_2, -P_3, \Omega_1, \Omega_2, \o'$ whose divisor is $(P_1) + (P_2) + (-P_3) + (\o') - 2(\Omega_1) -2(\Omega_2)$. Let $\ell_{1, P_3}$ be the horizontal line going through $P_3$ whose divisor is $div(\ell_{1, P_3}) = (P_3) + (-P_3) - 2(\Omega_2)$, and $\ell_{2, \o}$ is the vertical line going through $\o$ and $\o'$ whose divisor is $(\o) + (\o') - 2(\Omega_1)$. Then we have
\begin{equation}\label{Millerfunction}
div\left( \frac{\phi_{P_1, P_2}}{\ell_{1, P_3}\ell_{2, \o}} \right) \sim (P_1) + (P_2) - (P_3) - (\o).
\end{equation}
\end{theorem}

The rational function $g_{P_1, P_2} = \frac{\phi_{P_1, P_2}}{\ell_{1, P_3} \ell_{2, \o}}$ consisting of three terms, can be thus considered as Miller function on Edwards curves. Miller's algorithm for Edwards curves using this function works as in Algorithm~\ref{algo:Miller}.\\

\begin{algorithm}
   \KwIn{  $r = \sum_{i = 0}^t r_i 2^i$ with $r_i \in \{0, 1\}$, $P, Q \in E[r]$;}
   \KwOut{ $f = f_r(Q)$;}
  \BlankLine
   {
   $R \leftarrow P$, $f \leftarrow 1$, $g \leftarrow 1$\\
 \For{$i = t - 1$  \emph{\KwTo}  $0$}
 {
        $f \leftarrow {f}^2 \cdot \phi_{R, R}(Q)$ \\
	$g \leftarrow {g}^2 \cdot \ell_{1, \o}(Q) \cdot \ell_{2, 2R}(Q)$ \\
	$R \leftarrow 2R$ \\
   \If{$(r_i = 1)$}   {
    ${f} \leftarrow {f} \cdot \phi_{R, P}(Q)$ \\
    ${g} \leftarrow {g} \cdot \ell_{1, \o}(Q) \cdot \ell_{2, R + P}(Q)$\\
    $R \leftarrow R + P$ \\
   }  }
    \KwRet{ $f/g$ } }
   \caption{Miller's Algorithm for twisted Edwards curves~\cite{XL10}}
   \label{algo:Miller}

\end{algorithm}


\subsection{Xu-Lin Refinements}
For simplicity,  in what follows, we make use of the notation $\phi_{P, P}$ (resp. $\ell_{2,[2]P}$, and $\ell_{1, \o}$) replacing for $\phi_{P, P}(Q)$ (resp. $\ell_{2,[2]P}(Q)$, and $\ell_{1, \o}(Q)$). 
By extending the Blake et al.'s method~\cite{BMX06} to Edwards curves, Xu and Lin~\cite{XL10} presented a refinement to Miller algorithm. Their algorithm was achieved owing to the following theorem.

\begin{theorem}[Theorem 1 in~\cite{XL10}]
Let $E_{a, d}$ be a twisted Edwards curve over $\F_p$ and $P, R \in  E_{a, d}$ be a point of order $r$. Then\footnote{There were typos in the first formula of Theorem 1 in~\cite{XL10}. It should be $\ell_{2,[2]R}\ell_{1, \o}$ instead of $\ell_{1,[2]R}\ell_{2, \o}$.}

\begin{enumerate}
\item $$(\frac{\phi_{R, R}}{\ell_{2,[2]R}\ell_{1, \o}})^2 \frac{\phi_{[2]R, [2]R}}{\ell_{2,[4]R}\ell_{1, \o}} = \frac{\phi_{R, R}^2}{\phi_{[-2]R,[-2]R}\phi_{\o, \o}};$$

\item $$\frac{\phi_{R, R}}{\ell_{2,[2]R}\ell_{1, \o}} \frac{\phi_{[2]R, P}}{\ell_{2,[2]R + P}\ell_{1, \o}} = \frac{\phi_{R, R} \ell_{2, P}}{\phi_{[2]R + P,-P}\ell_{1, \o}}.$$

\end{enumerate}

\end{theorem}

The above theorem was proven by calculating divisors (see~\cite{XL10} for more details). From this theorem, they introduced refinements and an improved Miller algorithm in radix-4 representation~\cite[Algorithm 3]{XL10}. 
They also claimed that the total cost of the proposed algorithm is about $76.8\%$ of that of the original Miller's algorithm.

\section{Our Improvements on Miller's Algorithm}

\subsection{First Improvement}

We first present a new rational function whose divisor is equivalent to Miller function (Eq~\ref{Millerfunction}) presented in \cite{ALNR10}.

\begin{definition}\label{def:Millerfunction}
Let $E_{a,d}$ be a twisted Edwards curve and $R, P \in E_{a,d}$. Let $\phi_{R, P}$ be a conic passing through  $R$ and $P$, $\phi_{-R, -P}$ be a conic passing through $-R$ and $-P$, and let $\phi_{R + P, -[R + P]}$ a conic passing through $R + P$ and $-[R + P]$. Then we define

\begin{equation}\label{eq:Millerfunction}
g_{R, P} = \frac{\phi_{R, P}}{\phi_{R + P, -[R + P]}}.
\end{equation}

\end{definition}

\begin{lemma}\label{lem:1}

 We have 

$$div(g_{R, P}) = (R) + (P) - ([R + P]) - \o.$$
\end{lemma}

\begin{proof}
By calculating divisors, it is straightforward to see that 
$$\begin{array}{rcl}
\mbox{div}(g_{R, P}) &  = & (R) + (P) + (-[R + P]) + (\o') - 2(\Omega_1) -2(\Omega_2) \\
& & - ([R + P]) - (-[R + P]) - (\o) - (\o') + 2(\Omega_1) + 2(\Omega_2) \\
 & = & (R) + (P)  - ([R + P]) - (\o) \;. \\  
\end{array} $$ 
which concludes the proof.

\end{proof}

Recall that, on Weierstrass curves, the Miller function $g_{R, P} = \frac{\ell_{R, P}}{\upsilon_{R + P}}$, where $\ell_{R, P}, \upsilon_{R + P}$ are lines passing through $R, P$ and $R + P, -[R + P]$, respectively. The Eq.~(\ref{eq:Millerfunction}) thus looks similar to the Miller function on Weierstrass curves if a conic plays role as a line function. A variant of Miller algorithm by using Eq. 4 is described in Algorithm~\ref{algo:Miller1}.

\begin{algorithm}
   \KwIn{  $r = \sum_{i = 0}^t r_i 2^i$ with $r_i \in \{0, 1\}$, $P, Q \in E[r]$;}
   \KwOut{ $f = f_r(Q)$;}
  \BlankLine
   {
   $R \leftarrow P$, $f \leftarrow 1$, $g \leftarrow 1$\\
 \For{$i = t - 1$  \emph{\KwTo}  $0$}
 {
    $f \leftarrow {f}^2 \cdot \phi_{R, R}(Q)$ \\
	$g \leftarrow {g}^2  \cdot \phi_{2R, -2R}(Q)$ \\
	$R \leftarrow 2R$ \\
   \If{$(r_i = 1)$}   {
     ${f} \leftarrow {f} \cdot \phi_{R, P}(Q)$ \\
	$g \leftarrow {g} \cdot \phi_{R + P, -(R + P)}(Q)$ \\
	$R \leftarrow R + P$ \\
   }  }
    \KwRet{ $f/g$ } }
   \caption{First improvement of Miller's Algorithm for twisted Edwards curves}
   \label{algo:Miller1}

\end{algorithm}

\begin{remark}
As the original Miller's algorithm, our algorithm cannot avoid divisions needed to update $f$. But we can reduce them easily to one inversion at the end of the addition chain (for the cost of one squaring in addition at the each step of the algorithm).
\end{remark}

At first glance, Algorithm~\ref{algo:Miller1} requires only one multiplication for updating the function $g$ instead of two in Algorithm~\ref{algo:Miller}. Note that this operation is costly because it is performed in the full extension finite field. In Section~\ref{sec:discussion}, we will provide a detailed analysis on the performance of these algorithms. In the following section, we introduce a further refinement that even offers a better performance in comparison to Algorithm~\ref{algo:Miller1}.
 
\subsection{Refinement}
 
The new refinement is inspired from the following lemmas. 

\begin{lemma} 
\label{lem:2}
We have


\begin{equation}
g_{R, P} = \frac{\phi_{R, -R}\cdot \phi_{P, -P}}{\phi_{-R, -P}\cdot \phi_{\o,\o}}. 
\label{eq2}
\end{equation}

\end{lemma}

\begin{proof}

This lemma is again proved by considering divisors.
Indeed, $$\begin{array}{rcl}

\mbox{div}(\frac{\phi_{R, -R}\cdot \phi_{P, -P}}{\phi_{-R, -P}\cdot \phi_{\o,\o}}) & = & (R) + (-R) + (\o) + (\o') - 2(\Omega_1) -2(\Omega_2)\\
& & + (P) + (-P) + (\o) + (\o') - 2(\Omega_1) -2(\Omega_2)\\
& & - (-R) - (-P) - ([R + P]) -  (\o') + 2(\Omega_1) + 2(\Omega_2)\\
& & - 3(\o) -  (\o') + 2(\Omega_1) + 2(\Omega_2)\\
& = &  (R) + (P)  - ([R + P]) - (\o) \\
& = & \mbox{div}(g_{R, P}) \;. \\
\end{array} $$ 
which concludes the proof.

\end{proof}

\begin{lemma}\label{lem:3} Let $P, R \in E_{a, d}$ be points of order $r$, we have
\begin{equation}
\frac{\phi_{R, R}}{\phi_{R, -R}^2 \cdot  \phi_{2R, -2R}} = \frac{1}{\phi_{-R, -R}\cdot \phi_{\o, \o}}. 
\label{eq3}
\end{equation} 

\end{lemma}

This lemma is easy to be proven using Definition~\ref{def:Millerfunction} and Lemma~\ref{lem:2}. The factor $\phi_{\o, \o}$ can be precomputed and integrated into the factor $\phi_{-R, -R}$ as follows:

\begin{small}
\begin{align*}
\label{eq:simply1}
 \phi'_{-R, -R}(Q) & = \phi_{-R, -R}(Q) \cdot \phi_{\o, \o}(Q) \nonumber \\& = (c_{Z^2}(Z_Q^2 + Y_Q Z_Q) + c_{XY} X_Q Y_Q + c_{XZ} X_Q Z_Q)(X_Q(Z_Q - Y_Q)) \nonumber \\
			  & = c_{Z^2}\eta_1 + c_{XY} \eta_2 + c_{XZ} \eta_3,
\end{align*}
\end{small}

\noindent where $\eta_1 = (Z_Q^2 + Y_Q Z_Q)(X_Q(Z_Q - Y_Q))$, $\eta_2 =  X^2_Q Y_Q(Z_Q - Y_Q)$, $\eta_3 = X^2_Q Z_Q(Z_Q - Y_Q)$ are fixed for whole computation, thus they can be precomputed and stored. The Eq.~\ref{eq3} can thus be rewritten as follows:

\begin{equation}
\frac{\phi_{R, R}}{\phi_{R, -R}^2 \cdot  \phi_{2R, -2R}} = \frac{1}{\phi'_{-R, -R}}. 
\label{eq4}
\end{equation} 

Our algorithm is described by the pseudo-code in Algorithm~\ref{our:algoMiller} by applying Eq~\ref{eq4}. To do so, the proposed algorithm tries to delay the factor $\phi_{R, -R}$ for the next step. We make use of a memory variable $m$ to imply whether a delayed factor in the current step or not. If $m = 1$, there is a delayed factor for the next step and otherwise. There will be no delayed factor for the next step ({\it i.e.}, $m$ will be assigned to $0$) if the current bit $b_i = 0$ and there exists a delayed factor for the current step (line~\ref{eco2} in the Algorithm~\ref{our:algoMiller}). This will lead to the most expensive case (line~\ref{eco4}) of the proposed algorithm if the next bit $b_{i - 1} = 1$ (counting from the most significant bit).




\begin{algorithm}
   
   \KwIn{  $r = \sum_{i = 0}^t b_i 2^i$, $b_i \in \{0, 1\}$.}
   \KwOut{ $f$}
  \BlankLine
   { 
   $T \leftarrow P$, $f \leftarrow 1$, $g \leftarrow 1$, $m \leftarrow 0$ \;
 \For{$i = t - 1$  \emph{\KwTo}  $0$} 
 {  
    \nl \If{$(b_i = 0) \land (m = 0)$} { \label{eco1}
      $f \leftarrow f^2 \cdot \phi_{R,R} \; ; \quad$ $g \leftarrow g^2 \; ; \quad$ $R \leftarrow 2R \; ; \quad$\\
    }

    \nl \If{$(b_i = 0) \land (m = 1)$} { \label{eco2} 
      $f \leftarrow f^2 \; ; \quad$ $g \leftarrow g^2 \cdot \phi'_{-R,-R} \; ; \quad$ 
	$R \leftarrow 2R \; ; \quad$ \\
    }

    \nl \If{$(b_i = 1) \land (m = 1)$} { \label{eco3} 
      $f \leftarrow f^2 \cdot \phi_{2R,P} \; ; \quad$ $g \leftarrow g^2 \cdot \phi'_{-R,-R} \; ; \quad$ 
	$R \leftarrow 2R + P \; ; \quad$ \\
    }

    \nl \If{$(b_i = 1) \land (m = 0)$} { \label{eco4} 
      $f \leftarrow f^2 \cdot \phi_{R,R} \cdot \phi_{2R,P} \; ; \quad$ $g \leftarrow g^2 \cdot \phi'_{2R, -2R}  \; ; \quad$
	$T \leftarrow 2R + P \; ; \quad$\\
    }
$m \leftarrow \lnot m \vee b_i$ \\ 
	
   }
  
  \KwRet{ $\frac{f}{g}$} 
}

\caption{Improved Refinement of Miller's Algorithm for any Pairing-Friendly Edwards Curves}
\label{our:algoMiller}

 \end{algorithm}





\section{Discussion}\label{sec:discussion}
In this section, we first compare the proposed algorithm with the original Miller's algorithm over Edwards curves~\cite{ALNR10, XL10}, and the Xu-Lin refinements~\cite{XL10}. Before analyzing the costs of algorithm, we introduce notations for field arithmetic costs. Let $\F_{p^m}$ be an extension of degree $m$ of $\F_p$ for $m \ge 1$ and let ${\bf I}_{p^m}$, ${\bf M}_{p^m}$, ${\bf S}_{p^m}$, and ${\bf add}_{p^m}$ the costs for inversion, multiplication, squaring, and addition in the field $\F_{p^m}$, respectively. We denote by ${\bf m}_a$ the multiplication by the curve coefficient $a$. 

The cost of the algorithms for pairing computation consists of three parts: the cost of updating the functions $f, g$, the cost of updating the point $R$ and the cost of evaluating rational functions at some point $Q$. Note that during Ate pairing computation, coordinates of the point $R$ that is on the twisted curve, . The analysis in~\cite{ALNR10} showed that the total cost of updating the point $R$ and coefficients $c_{Z^2}, c_{XY}$, and $c_{ZZ}$ of the conic is $6{\bf M}_{p^e}+ 5{\bf S}_{p^e} + 2{\bf m}_a$ for each doubling step and $14{\bf M}_{p^e} + 1{\bf m}_a$ for each addition step (see~\cite[\S 5]{ALNR10} for more details), where $e = k/d$. 
Without special treatment, this cost is the same for all algorithms.

The most costly operations in pairing computations are operations in the full extension field $\F_{p^k}$. At high levels of security (i.e. $k$ large), the complexity of operations in $\F_{p^k}$ dominates the complexity of the operations that occur in the lower degree subfields. In this subsection, we only analyze the cost of updating the functions $f, g$ which are generally executed on the full extension field $\mathbb{F}_{p^k}$. Assume that the ratio of one full extension field squaring to one full extension field multiplication is set to ${\bf S}_{p^k} = 0.8 {\bf M}_{p^k}$, a commonly used value in the literature. 
It is clear to see that to update functions $f$ and $g$, the proposed algorithm requires $1{\bf M}_{p^k} + 2{\bf S}_{p^k}$  for a doubling step (lines 1, 3), and $1{\bf M}_{p^k}$ for an addition step (lines 2, 4). TABLE~\ref{tab:Cmp} shows the number of operations needed in $\F_{p^k}$ for updating $f, g$ in different algorithms.

\begin{table*}
	\centering
		\begin{tabular}{|l|c|c|}
		\hline
			& Doubling & Doubling and Addition \\
		\hline
		{\bf Algorithm~\ref{algo:Miller}} & $2{\bf S}_{p^k} + 3{\bf M}_{p^k}$ & $2{\bf S}_{p^k} + 5{\bf M}_{p^k}$\\
		(Miller's algorithm~\cite{ALNR10, XL10})&	= $4.6{\bf M}_{p^k}$ &	= $6.6{\bf M}_{p^k}$	\\
		\hline
		\multirow{2}{*}{{\bf Algorithm in~\cite{ALNR10}}} &  $1{\bf S}_{p^k} + 1{\bf M}_{p^k}$ & $1{\bf S}_{p^k} + 2{\bf M}_{p^k}$\\
						 &	= $1.8{\bf M}_{p^k}$	&	= $2.8{\bf M}_{p^k}$ \\
		\hline
		\multirow{2}{*}{{\bf Algorithm~\ref{algo:Miller1}}} &  $2{\bf S}_{p^k} + 2{\bf M}_{p^k}$ & $2{\bf S}_{p^k} + 4{\bf M}_{p^k}$\\
						 &	= $3.6{\bf M}_{p^k}$	&	= $5.6{\bf M}_{p^k}$ \\
		\hline

		\multirow{2}{*}{{\bf Algorithm~\ref{our:algoMiller}}} & $ 2{\bf S}_{p^k} + 1{\bf M}_{p^k}$ & $ 2{\bf S}_{p^k} + 2{\bf M}_{p^k} = 3.6{\bf M}_{p^k}$(line {\bf 3}) \\
		& = $2.6{\bf M}_{p^k}$ & $ 2{\bf S}_{p^k} + 3{\bf M}_{p^k} = 4.6{\bf M}_{p^k}$ (line {\bf 4})\\
									
		\hline
\end{tabular}
	\caption{Comparison of the cost of updating $f, g$ of Algorithms. ``Doubling'' is when algorithms deal with the bit ``$b_i = 0$'' and ``Doubling and Addition'' is when algorithms deal with the bit ``$b_i = 1$''.}
	\label{tab:Cmp}
\end{table*}

From Table~\ref{tab:Cmp}, for the \emph{generic} case we can see that {\bf Algorithm~\ref{our:algoMiller}} saves two full extension field multiplication when the bit $b_i = 0$ compared with {\bf Algorithm~\ref{algo:Miller}}. When the bit $b_i = 1$, {\bf Algorithm~\ref{our:algoMiller}} saves two or three full extension field multiplications in comparison to {\bf Algorithm~\ref{algo:Miller}}, depending on which case Algorithm~\ref{our:algoMiller} executes. 


In comparison to Arene {\em et al.}'s algorithm~\cite{ALNR10}, Algorithm~\ref{our:algoMiller} requires one more squaring in the full extension field for each doubling step. However, as already mentioned, Arene {\em et al.} can only be applied on Edwards curves with an even embedding degree $k$ for Tate pairing computation, while our approach is generic. It can be applied to any (pairing-friendly) Edwards curve and for both the Weil and the Tate pairing. 

The refinements in~\cite{XL10} are described in radix 4. Their algorithm allows to eliminate some rational functions from Eq~(\ref{Millerfunction}) during pairing computation. Let  $r = \sum_{i = 0}^{l' - 1} q_i 4^i$, with $q_i \in \{0, 1, 2, 3\}$. Table~\ref{tab:Cmp2} compares our algorithm and their algorithm. From Table~\ref{tab:Cmp2}, it clearly see that {\bf Algorithm~\ref{our:algoMiller}} is generally faster than the refinements of Miller's algorithm in~\cite{XL10} for all four cases.

\begin{table*}[h]
	\centering

		\begin{tabular}{|c|c|c|}
		\hline
		& {\bf Algorithm in~\cite{XL10}} & {\bf Algorithm~\ref{our:algoMiller}} \\
		\hline
		$q = 0$ & $5{\bf S}_{p^k} + 3{\bf M}_{p^k}$ & $4{\bf S}_{p^k} + 2{\bf M}_{p^k}$ \\
		\hline
		\multirow{2}{*}{$q = 1$}  & \multirow{2}{*}{$4{\bf S}_{p^k} + 7{\bf M}_{p^k}$} & $4{\bf S}_{p^k} + 3{\bf M}_{p^k}$ (line~{\bf 3}) \\
			& & $4{\bf S}_{p^k} + 4{\bf M}_{p^k}$ (line~{\bf 4})\\
		\hline
		\multirow{2}{*}{$q = 2$}  & \multirow{2}{*}{$4{\bf S}_{p^k} + 7{\bf M}_{p^k}$} & $4{\bf S}_{p^k} + 3{\bf M}_{p^k}$ (line~{\bf 3}) \\
		& & $4{\bf S}_{p^k} + 4{\bf M}_{p^k}$ (line~{\bf 4})\\
		\hline
		\multirow{2}{*}{$q = 3$}  & \multirow{2}{*}{$4{\bf S}_{p^k} + 10{\bf M}_{p^k}$} & $4{\bf S}_{p^k} + 4{\bf M}_{p^k}$ (line~{\bf 3}) \\
		& & $4{\bf S}_{p^k} + 5{\bf M}_{p^k}$ (line~{\bf 4})\\
		\hline
\end{tabular}
	\caption{Comparison of our algorithm with the refinements in~\cite{XL10}.}
	\label{tab:Cmp2}
\end{table*}

\section{Conclusion}
In this paper, we extended the Blake-Murty-Xu's method to propose further refinements to Miller's algorithm over Edwards curves. Our algorithm is \emph{generically} more efficient than the original Miller's algorithm (Algorithm~\ref{algo:Miller}) the Xu-Lin's refinements in~\cite{XL10}. Especially, the proposed algorithm can be applied for computing pairings of any type over any pairing-friendly Edwards curve. This allows the use of Edwards curves with embedding degree not of the form $2^i3^j$, and is suitable for the computation of {\em optimal pairings}~\cite{Ver10}.



\end{document}